\documentclass[a4paper]{article}

\usepackage{amsmath,amsfonts,amssymb,amsthm,graphicx,graphics,epsf}
\usepackage{graphicx}
\usepackage[numbers]{natbib}

\usepackage{algorithmic}
\usepackage{algorithm}
\usepackage{hyperref,color}
\usepackage{xspace}
\usepackage{lineno}

\usepackage{geometry}

\newtheorem{theorem}{Theorem}[section]

\newtheorem{obs}[theorem]{Observation}

\newtheorem{question}[theorem]{Question}
\newtheorem{conj}[theorem]{Conjecture}

\theoremstyle{remark}
\newtheorem*{remark}{Remark}

\title{Sharing a pizza: bisecting masses with two cuts}

\author{Luis Barba\thanks{Department of Computer Science,
        ETH Z\"{u}rich, Switzerland. {\tt \{luis.barba, alexander.pilz, patrick.schnider\}@inf.ethz.ch}}\ \ \thanks{Partially supported by the ETH Postdoctoral Fellowship}
        \and
        Alexander Pilz$^*$\ \ \thanks{Supported by an Erwin Schr\"odinger fellowship of the Austrian Science 
Fund (FWF): J-3847-N35}
        \and
        Patrick Schnider$^*$}

\begin{document}
\maketitle
\begin{abstract}
Assume you have a pizza consisting of four ingredients (e.g., bread, tomatoes, cheese and olives) that you want to share with your friend. You want to do this fairly, meaning that you and your friend should get the same amount of each ingredient. How many times do you need to cut the pizza so that this is possible? We will show that two straight cuts always suffice. More formally, we will show the following extension of the well-known Ham-sandwich theorem: Given four mass distributions in the plane, they can be simultaneously bisected with two lines. That is, there exist two oriented lines with the following property: let $R^+_1$ be the region of the plane that lies to the positive side of both lines and let $R^+_2$ be the region of the plane that lies to the negative side of both lines. Then $R^+=R^+_1\cup R^+_2$ contains exactly half of each mass distribution.
\end{abstract}

\section{Introduction}

The famous \emph{Ham-sandwich theorem} (see e.g.\ \cite{Matousek, StoneTukey}) states that any $d$ mass distributions in $\mathbb{R}^d$ can be simultaneously bisected by a hyperplane. In particular, a two-dimensional sandwich consisting of bread and ham can be cut with one straight cut in such a way that each side of the cut contains exactly half of the bread and half of the ham. However, if two people want to share a pizza, this result will not help them too much, as pizzas generally consist of more than two ingredients. There are two options to overcome this issue: either they don't use a straight cut, but cut along some more complicated curve, or they cut the pizza more than once. In this paper we investigate the latter option. In particular we show that a pizza with four ingredients can always be shared fairly using two straight cuts. See Figure \ref{Fig:Pizza} for an example.

\begin{figure}
\label{Fig:Pizza}
\centering 
\includegraphics{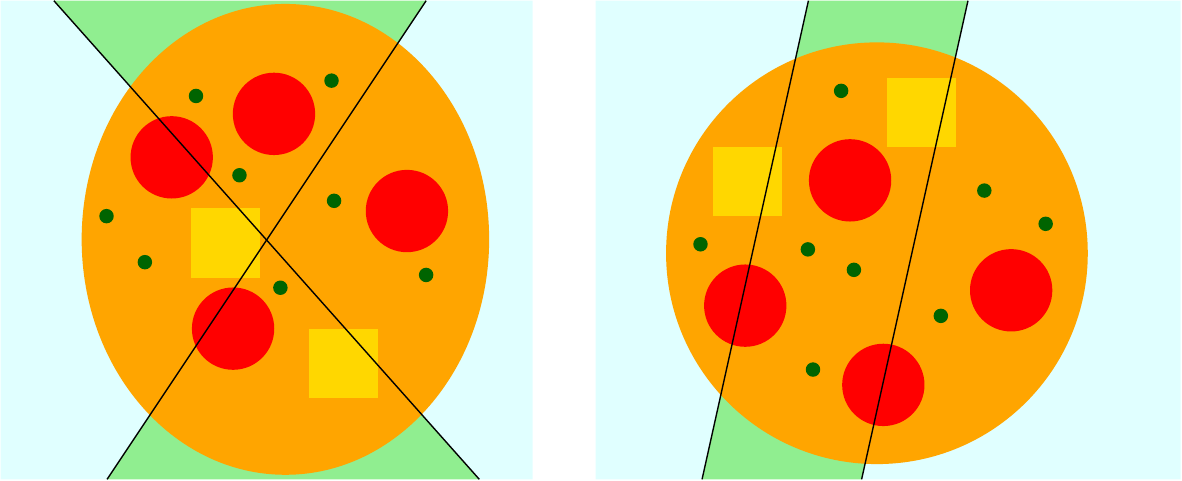}
\caption{Sharing a (not necessarily round) pizza fairly with two cuts. One person gets the parts in the light blue region, the other person gets the parts in the green region.}
\end{figure}

To phrase it in mathematical terms, we show that four mass distributions in the plane can be simultaneously bisected with two lines. A precise definition of what bisecting with $n$ lines means is given in the Preliminaries. This main result is proven in Section \ref{Sec:2Cuts}. In Section \ref{Sec:Restrictions} we add more restrictions on the lines.
In Section \ref{Sec:Algo} we give algorithms for all our results.
Finally, in Section \ref{Sec:General} we look at the general case of bisecting mass distributions in $\mathbb{R}^d$ with $n$ hyperplanes, and show an upper bound of $nd$ mass distributions that can be simultaneously bisected this way. We conjecture that this bound is tight, that is, that any $nd$ mass distributions in $\mathbb{R}^d$ can be simultaneously bisected with $n$ hyperplanes. For $d=1$, this is the well-known \emph{Necklace splitting problem}, for which an affirmative answer to our conjecture is known \cite{Hobby, Matousek}. So, our general problem can be seen as both a generalization of the Ham-sandwich theorem for more than one hyperplane, as well as a generalization of the Necklace splitting problem to higher dimensions.

Further, our results add to a long list of results about partitions of mass distributions, starting with the already mentioned Ham-sandwich theorem. A generalization of this is the polynomial Ham-sandwich theorem, which states that any $\binom{n+d}{d}-1$ mass distributions in $\mathbb{R}^d$ can be simultaneously bisected by an algebraic surface of degree $n$ \cite{StoneTukey}. Applied to the problem of sharing a pizza, this result gives an answer on how complicated the cut needs to be, if we want to use only a single (possibly self-intersecting) cut.

While bisections with several lines were studied by several authors, this particular question was to our knowledge first introduced by Bereg et al \cite{Bereg}, who showed that three point sets can always be simultaneously bisected with two lines. 
In this paper, we provide a substantial strengthening of their result in two ways: (1) instead of point sets, we generalize the results to work with mass distributions; and (2) we show that, in fact, a fourth mass distribution can also be simultaneously bisected (Section \ref{Sec:2Cuts}), or we can use this extra degree of freedom to put more restrictions on the bisecting lines (Section \ref{Sec:Restrictions}).
For example, we can find a bisection of three mass distributions with two lines, where one of the lines is required to pass through a given point in the plane, or it is required to be parallel to a given line. 

Several results are also known about equipartitions of mass distributions into more than two parts. A straightforward application of the 2-dimensional Ham-sandwich theorem is that any mass distribution in the plane can be partitioned into four equal parts with 2 lines. It is also possible to partition a mass distribution in $\mathbb{R}^3$ into 8 equal parts with three planes, but for $d\geq 5$, it is not always possible to partition a mass distribution into $2^d$ equal parts using $d$ hyperplanes \cite{Edelsbrunner}. The case $d=4$ is still open. A result by Buck and Buck \cite{Buck} states that a mass distribution in the plane can be partitioned into 6 equal parts by 3 lines passing through a common point. Several results are known about equipartitions in the plane with \emph{k-fans}, i.e., $k$ rays emanating from a common point. Note that 3 lines going through a common point can be viewed as a 6-fan, thus the previously mentioned result shows that any mass partition in the plane can be equipartitioned by a 6-fan. Motivated by a question posed by Kaneko and Kano \cite{Kaneko}, several authors have shown independently that 2 mass distributions in the plane can be simultaneously partitioned into 3 equal parts by a 3-fan \cite{Kirkpatrick, Ito, Sakai}. The analogous result for 4-fans holds as well \cite{Barany}. Partitions into non-equal parts have also been studied \cite{Zivaljevic}. All these results give a very clear description of the sets used for the partitions. If we allow for more freedom, much more is possible. In particular, Sober{\'o}n \cite{Soberon} and Karasev \cite{Karasev} have recently shown independently that any $d$ mass distributions in $\mathbb{R}^d$ can be simultaneously equipartitioned into $k$ equal parts by $k$ convex sets. The proofs of all of the above mentioned results rely on topological methods, many of them on the famous Borsuk-Ulam theorem and generalizations of it. For a deeper overview of these types of arguments, we refer to Matou{\v{s}}ek's excellent book \cite{Matousek}.

\subsection*{Preliminaries}
A \emph{mass distribution} $\mu$ on $\mathbb{R}^d$ is a measure on $\mathbb{R}^d$ such that all open subsets of $\mathbb{R}^d$ are measurable, $0<\mu(\mathbb{R}^d)<\infty$ and $\mu(S)=0$ for every lower-dimensional subset $S$ of $\mathbb{R}^d$.
Let $\mathcal{L}$ be a set of oriented hyperplanes. For each $\ell\in\mathcal{L}$, let $\ell^+$ and $\ell^-$ denote the positive and negative side of $\ell$, respectively (we consider the sign resulting from the evaluation of a point in these sets into the linear equation defining $\ell$). We explicitly allow hyperplanes at infinity, for which $\ell^+$ and $\ell^-$ are all of $\mathbb{R}^d$ and the empty set, respectively. For every point $p\in\mathbb{R}^d$, define $\lambda(p):=|\{\ell\in\mathcal{L}\mid p\in\ell^+\}|$ as the number of hyperplanes that have $p$ in their positive side. Let $R^+:=\{p\in\mathbb{R}^d\mid \lambda(p) \text{ is even}\}$ and $R^-:=\{p\in\mathbb{R}^d\mid \lambda(p) \text{ is odd}\}$. We say that $\mathcal{L}$ \emph{bisects} a mass distribution $\mu$ if $\mu(R^+)=\mu(R^-)$. For a family of mass distributions $\mu_1,\ldots,\mu_k$ we say that $\mathcal{L}$ \emph{simultaneously bisects} $\mu_1,\ldots,\mu_k$ if $\mu_i(R^+)=\mu_i(R^-)$ for all $i\in\{1,\ldots,k\}$.

More intuitively, this definition can also be understood the following way: if $C$ is a cell in the hyperplane arrangement induced by $\mathcal{L}$ and $C'$ is another cell sharing a facet with $C$, then $C$ is a part of $R^+$ if and only if $C'$ is a part of $R^-$. See Figure \ref{Fig:Checkers} for an example.

\begin{figure}
\centering 
\includegraphics{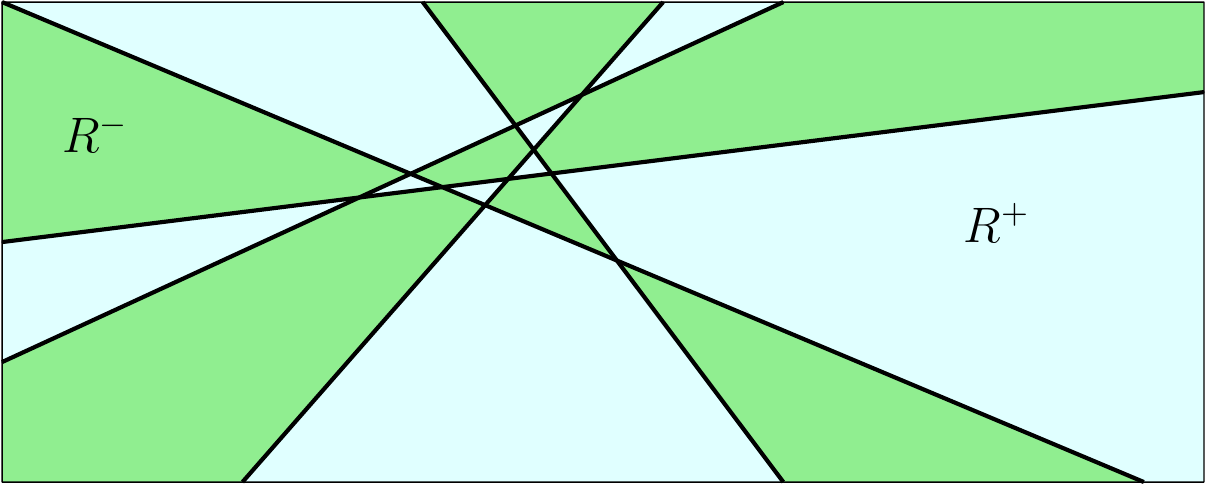}
\caption{The regions $R^+$ (light blue) and $R^-$ (green).}
\label{Fig:Checkers}
\end{figure}

Let $g_i(x):=a_{i,1}x_1+\ldots+a_{i,d}x_d+a_{i,0}\geq 0$ be the linear equation describing $\ell_i^+$ for $\ell_i\in\mathcal{L}$. Note that a hyperplane at infinity is described by an equation of the form $g_i(x)=a_{i,0}\geq 0$. The following is yet another way to describe $R^+$ and $R^-$: a point $p\in\mathbb{R}^d$ is in $R^+$ if $\prod_{\ell_i\in\mathcal{L}}g_i(p)\geq 0$ and it is in $R^-$ if $\prod_{\ell_i\in\mathcal{L}}g_i(p)\leq 0$. That is, if we consider the union of the hyperplanes in $\mathcal{L}$ as an oriented algebraic surface of degree $|\mathcal{L}|$, then $R^+$ is the positive side of this surface and $R^-$ is the negative side.

Note that reorienting one line just maps $R^+$ to $R^-$ and vice versa. In particular, if a set $\mathcal{L}$ of oriented hyperplanes simultaneously bisects a family of mass distributions $\mu_1,\ldots,\mu_k$, then so does any set $\mathcal{L'}$ of the same hyperplanes with possibly different orientations. Thus we can ignore the orientations and say that a set $\mathcal{L}$ of (undirected) hyperplanes simultaneously bisects a family of mass distributions if some orientation of the hyperplanes does.

\section{Two Cuts}
\label{Sec:2Cuts}

In this section we will look at simultaneous bisections with two lines in $\mathbb{R}^2$ and with two planes in $\mathbb{R}^3$. Both proofs rely on the famous Borsuk-Ulam theorem \cite{Borsuk}, which we will use in the version of \emph{antipodal mappings}. An antipodal mapping is a continuous mapping $f: S^d\rightarrow\mathbb{R}^d$ such that $f(-x)=-f(x)$ for all $x\in S^d$.

\begin{theorem}[Borsuk-Ulam theorem~\cite{Matousek}]
For every antipodal mapping $f: S^d\rightarrow\mathbb{R}^d$ there exists a point $x\in S^d$ satisfying $f(x)=0$.
\end{theorem}

The proof of the Ham-sandwich theorem can be derived from the Borsuk-Ulam theorem in the following way (see, e.g.,~\cite{Matousek}).
Let $\mu_1$ and $\mu_2$ be two mass distributions in $\mathbb{R}^2$.
For a point $p = (a, b, c) \in S^2$, consider the equation of the line $ax + by + c = 0$ and note that it defines a line in the plane parameterized by the coordinates of $p$.
Moreover, it splits the plane into two regions, the set $R^+(p) = \{(x,y)\in \mathbb{R}^2 : ax + by + c \geq 0\}$ and the set $R^-(p) = \{(x,y)\in \mathbb{R}^2 : ax + by + c \leq 0\}$.
Thus, we can define two functions $f_i := \mu_i(R^+(p)) - \mu_i(R^-(p))$ that together yield a function $f:S^2 \to \mathbb{R}^2$ that is continuous and antipodal.
Thus, by the Borsuk-Ulam theorem, there is a point $p = (a, b, c) \in S^2$, such that $f_i(-p) = -f_i(p)$ for $i\in \{1, 2\}$, which implies that the line $ax + by + c = 0$ defined by $p$ is a Ham-sandwich cut.
In this paper, we use variants of this proof idea to obtain simultaneous bisections by geometric objects that are parameterized by points in $S^d$. 
The main difference is that we replace some of the $f_i$'s by other functions, whose vanishing enforces specific structural properties on the resulting bisecting object.
We are now ready to prove our first main result:

\begin{theorem}
\label{Thm:2Lines}
Let $\mu_1,\mu_2,\mu_3,\mu_4$ be four mass distributions in $\mathbb{R}^2$. Then there exist two lines $\ell_1,\ell_2$ such that $\{\ell_1,\ell_2\}$ simultaneously bisects $\mu_1,\mu_2,\mu_3,\mu_4$.
\end{theorem}

\begin{proof}
For each $p=(a,b,c,d,e,g)\in S^5$ consider the bivariate polynomial $c_p(x,y)=ax^2+by^2+cxy+dx+ey+g$. Note that $c_p(x,y)=0$ defines a conic section in the plane. Let $R^+(p):=\{(x,y)\in\mathbb{R}^2\mid c_p(x,y)\geq 0\}$ be the set of points that lie on the positive side of the conic section and let $R^-(p):=\{(x,y)\in\mathbb{R}^2\mid c_p(x,y)\leq 0\}$ be the set of points that lie on its negative side. Note that for $p=(0,0,0,0,0,1)$ we have $R^+(p)=\mathbb{R}^2$ and $R^-(p)=\emptyset$, and vice versa for $p=(0,0,0,0,0,-1)$. Also note that $R^+(-p)=R^-(p)$. We now define four functions $f_i:S^5\rightarrow\mathbb{R}$ as follows: for each $i\in\{1,\ldots,4\}$ define $f_i:=\mu_i(R^+(p))-\mu_i(R^-(p))$. From the previous observation it follows immediately that $f_i(-p)=-f_i(p)$ for all $i\in\{1,\ldots,4\}$ and $p\in S^5$. It can also be shown that the functions are continuous, but for the sake of readability we postpone this step to the end of the proof. Further let
\[
A(p):=\det
 \begin{pmatrix}
  a & c/2 & d/2 \\
  c/2& b & e/2 \\
  d/2 & e/2 & g
 \end{pmatrix}.
\]
It is well-known that the conic section $c_p(x,y)=0$ is degenerate if and only if $A(p)=0$. Furthermore, being a determinant of a $3\times 3$-matrix, $A$ is continuous and $A(-p)=-A(p)$. Hence, setting $f_5(p):=A(p)$, $f:=(f_1,\ldots,f_5)$ is an antipodal mapping from $S^5$ to $\mathbb{R}^5$, and thus by the Borsuk-Ulam theorem, there exists $p^*$ such that $f(p^*)=0$. For each $i\in\{1,\ldots,4\}$ the condition $f_i(p^*)=0$ implies by definition that $\mu_i(R^+(p^*))=\mu_i(R^-(p^*))$. The condition $f_5(p^*)=0$ implies that $c_p(x,y)=0$ describes a degenerate conic section, i.e., two lines (possibly one of them at infinity), a single line of multiplicity 2, a single point or the empty set. For the latter three cases, we would have $R^+(p^*)=\mathbb{R}^2$ and $R^-(p^*)=\emptyset$ or vice versa, which would contradict $\mu_i(R^+(p^*))=\mu_i(R^-(p^*))$. Thus $f(p^*)=0$ implies that $c_p(x,y)=0$ indeed describes two lines that simultaneously bisect $\mu_1,\mu_2,\mu_3,\mu_4$.

It remains to show that $f_i$ is continuous for $i\in\{1,\ldots,4\}$. To that end, we will show that $\mu_i(R^+(p))$ is continuous. The same arguments apply to $\mu_i(R^-(p))$, which then shows that $f_i$, being the difference of two continuous functions, is continuous. So let $(p_n)_{n=1}^{\infty}$ be a sequence of points in $S^5$ converging to $p$. We need to show that $\mu_i(R^+(p_n))$ converges to $\mu_i(R^+(p))$. If a point $q$ is not on the boundary of $R^+(p)$, then for all $n$ large enough we have $q\in R^+(p_n)$ if and only if $q\in R^+(p)$. As the boundary of $R^+(p)$ has dimension 1 and $\mu_i$ is a mass distribution we have $\mu_i(\partial R^+(p))=0$ and thus $\mu_i(R^+(p_n))$ converges to $\mu_i(R^+(p))$ as required.
\end{proof}

\begin{remark}
In the conference version of this paper \cite{conference}, a similar result in $\mathbb{R}^3$ was also claimed. However, there was an error in the proof.
\end{remark}

We note that our setting allows for a consistent formulation in the real projective plane.
Indeed, we required the line at infinity as a special case.
Further, while a single line does not bisect the projective plane, it is bisected by two lines in exactly the two regions that we considered.
However, as it is convenient to have a notion of orientation of a separator, our presentation will use the affine plane and treat the line at infinity as a special case.

\section{Putting more restrictions on the cuts}
\label{Sec:Restrictions}

In this section, we look again at bisections with two lines in the plane. 
However, we enforce additional conditions on the lines, at the expense of reducing the number of bisected mass distributions.

\begin{theorem}
\label{Thm:Parallel}
Let $\mu_1,\mu_2,\mu_3$ be three mass distributions in $\mathbb{R}^2$. Given any line $\ell$ in the plane, there exist two lines $\ell_1,\ell_2$ such that $\{\ell_1,\ell_2\}$ simultaneously bisects $\mu_1,\mu_2,\mu_3$ and $\ell_1$ is parallel to~$\ell$.
\end{theorem}

We note here that a line at infinity is parallel to any other line.

\begin{proof}
Assume without loss of generality that $\ell$ is parallel to the $x$-axis; otherwise rotate $\mu_1,\mu_2,\mu_3$ and $\ell$ to achieve this property. Consider the conic section defined by the polynomial $ax^2+by^2+cxy+dx+ey+g$. If $a=0$ and the polynomial decomposes into linear factors, then one of the factors must be of the form $\beta y+\gamma$.
If $\beta=0$, then this factor corresponds to a line at infinity, which is by definition parallel to any line, thus also the $x$-axis.
Otherwise, the line defined by this factor is parallel to the $x$-axis. Thus, we can modify the proof of Theorem \ref{Thm:2Lines} in the following way: we define $f_1,f_2,f_3$ and $f_5$ as before, but set $f_4:=a$. It is clear that $f$ still is an antipodal mapping and thus has a zero.
As $f_5=0$, we know, like above, that the polynomial defining the conic section decomposes into linear factors, i.e., two lines.
Since $f_4=a=0$, we thus also know that one of the two corresponding lines is parallel to the $x$-axis.
Finally, $f_1=f_2=f_3=0$ implies that the two lines simultaneously bisect the three mass distributions, which proves the result.
\end{proof}

Another natural condition on a line is that it has to pass through a given point.

\begin{theorem}
\label{Thm:1throughPoint}
Let $\mu_1,\mu_2,\mu_3$ be three mass distributions in $\mathbb{R}^2$ and let $q$ be a point. Then there exist two lines $\ell_1,\ell_2$ such that $\{\ell_1,\ell_2\}$ simultaneously bisects $\mu_1,\mu_2,\mu_3$ and $\ell_1$ goes through $q$.
\end{theorem}

\begin{proof}
Assume without loss of generality that $q$ coincides with the origin; otherwise translate $\mu_1,\mu_2,\mu_3$ and $q$ to achieve this. 
Consider  again the conic section defined by the polynomial $ax^2+by^2+cxy+dx+ey+g$. If $g=0$ and the polynomial decomposes into linear factors, then one of the factors must be of the form $\alpha x+\beta y$. In particular, the line defined by this factor goes through the origin. We can again modify the proof of Theorem \ref{Thm:2Lines} in the following way: we define $f_1,f_2,f_3$ and $f_5$ as before, but set $f_4:=g$. It is clear that $f$ still is an antipodal mapping. The zero of this mapping now implies the existence of two lines simultaneously bisecting three mass distributions, one of them going through the origin, which proves the result.
\end{proof}

Note that the above results complement each other when considering the projective plane:
Theorem~\ref{Thm:1throughPoint} makes $\ell_1$ pass through a point $q$ in the affine plane, and Theorem~\ref{Thm:Parallel} covers the case in which $q$ is on the line at infinity.

At the cost of another mass distribution, we can also enforce the intersection of the two lines to be at a given point.

\begin{theorem}
\label{Thm:2throughPoint}
Let $\mu_1,\mu_2$ be two mass distributions in $\mathbb{R}^2$ and let $q$ be a point. Then there exist two lines $\ell_1,\ell_2$ such that $\{\ell_1,\ell_2\}$ simultaneously bisects $\mu_1,\mu_2$, and both $\ell_1$ and $\ell_2$ go through $q$.
\end{theorem}

\begin{proof}
Assume without loss of generality that $q$ coincides with the origin; otherwise translate $\mu_1,\mu_2$ and $q$ to achieve this.  
Consider the conic section defined by the polynomial $ax^2+by^2+cxy$, i.e., the conic section where $d=e=g=0$. If this conic section decomposes into linear factors, both of them must be of the form $\alpha x+\beta y=0$. In particular, both of them pass through the origin. Furthermore, as $d=e=g=0$, the determinant $A(p)$ vanishes, which implies that the conic section is degenerate. Thus, we can modify the proof of Theorem \ref{Thm:2Lines} in the following way: we define $f_1,f_2$ as before, but set $f_3:=d$, $f_4:=e$ and $f_5:=g$. It is clear that $f$ still is an antipodal mapping. The zero of this mapping now implies the existence of two lines simultaneously bisecting two mass distributions, both of them going through the origin, which proves the result.
\end{proof}

\section{Algorithmic aspects}
\label{Sec:Algo}
Going back to the planar case, instead of considering four mass distributions $\mu_1, \ldots, \mu_4$, one can think of having four finite sets of points $P_1, \ldots, P_4\subset \mathbb{R}^2$.
In this setting, our problem translates to finding two lines that simultaneously bisect these four point sets (that is, both (open) regions contain at most $\lfloor\frac{|P_i|}{2}\rfloor$ of the points of $P_i$). The existence of such a bisection follows from Theorem~\ref{Thm:2Lines} as we can always replace each point by a sufficiently small disk and consider their area as a mass distribution.

An interesting question is then to find efficient algorithms to compute such a bisection given any four sets $P_1, \ldots, P_4$ with a total of $n$ points.
For example, there exists a linear-time algorithm for Ham-sandwich cuts of two sets of points in~$\mathbb{R}^2$~\cite{lo1994algorithms}.
For the problem at hand, a trivial $O(n^5)$ time algorithm can be applied by looking at all pairs of combinatorially different lines. While this running time can be reduced using known data structures, it still goes through $\Theta(n^4)$ different pairs of lines. 
An algorithm that does not consider all combinatorially different pairs of lines is described in the proof of the following theorem.

\begin{theorem}
\label{Alg:2D}
Given any four planar point sets $P_1, \ldots, P_4$ with a total of $n$ points, one can find two lines $\ell_1,\ell_2$ such that $\{\ell_1,\ell_2\}$ simultaneously bisects $P_1, \ldots, P_4$ in $O(n^{\frac{10}{3}})$ time.
\end{theorem}
\begin{proof}
We know from Theorem \ref{Thm:2Lines} that a solution exists. Given a solution, we can move one of the lines to infinity using a projective transformation. After this transformation, the remaining line simultaneously bisects the four transformed point sets. In other words, given any four planar point sets $P_1, \ldots, P_4$, we can always find a projective transformation $\phi$ such that $\phi(P_1), \ldots, \phi(P_4)$ can be simultaneously bisected by a single line.
Checking whether four point sets can be simultaneously bisected by a line can be done by first building the dual line arrangement of the union of the four sets in $O(n^2)$ time~\cite{dual_arr0,dual_arr1} (where a point $(a,b)$ is replaced by the line $y=ax+b$ and vice versa).
We can then walk along the middle level of the arrangement, keeping track of how many of the dual lines of each point set are above and below the middle level, which tells us whether somewhere along the middle level exactly half of the dual lines of every point set are above.
For the starting point of our walk, we count the number of dual lines above and below the middle level in linear time and every update only needs constant time.
Thus, the time needed after building the arrangement is bounded by the complexity of the middle level, which is at most $O(n^{\frac{4}{3}})$, as shown by Dey~\cite{Dey}.
The choice of the line at infinity for a projective transformation of a point set corresponds to choosing the north pole (i.e., the point at vertical infinity, which is dual to the line at infinity) in the dual.
The north pole is contained in one of the $O(n^2)$ cells of the dual arrangement.
So in order to check for every possible projective transformation~$\phi$ whether $\phi(P_1), \ldots, \phi(P_4)$ can be simultaneously bisected by a line, it suffices to build the dual arrangement once;
after that, we can check whether $\phi(P_1), \ldots, \phi(P_4)$ can be simultaneously bisected by a line for every combinatorially different choice of the line at infinity in time $O(n^{\frac{4}{3}})$ per choice.
See \figurename~\ref{fig:middle_levels}.
As there are $O(n^2)$ combinatorially different choices for the line at infinity of a projective transformations (i.e., cells in the dual arrangement), the running time of $O(n^{2+\frac{4}{3}})=O(n^{\frac{10}{3}})$ follows.
\end{proof}

\begin{figure}
\centering
\includegraphics{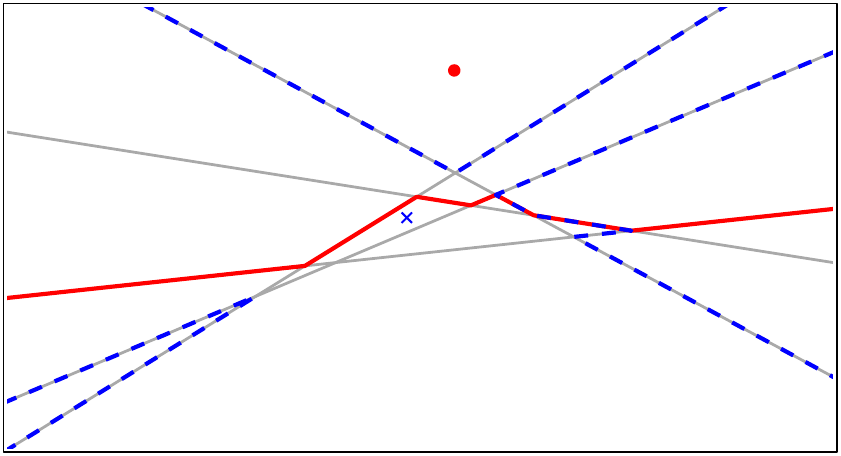}
\caption{
A dual line arrangement.
The red round dot marks the cell containing the point at vertical infinity, which results in the middle level indicated in bold red.
When choosing the blue cross as the north pole, the middle level is indicated by the dashed blue segments.
Note that these form a connected cycle in the projective plane, and that we can thus re-use the initially computed line arrangement.
}
\label{fig:middle_levels}
\end{figure}

\sloppypar{
The analysis of the above algorithm heavily depend on Dey's result~\cite{Dey} on the middle level in arrangements.
The current best lower bound on the complexity of the middle level is $\Omega(n \log n)$~\cite{level_lower_bound}.
Note that in the analysis of our algorithm we implicitly use an upper bound of $O(n^{\frac{10}{3}})$ for the complexity of all projectively different middle levels.
More formally, let $c$ be a cell in the dual line arrangement $\mathcal{A}$ and let $m(c)$ be the complexity of the middle level when the north pole lies in $c$.
Then $\sum_{c\in\mathcal{A}}m(c)$ is upper bounded by $O(n^{\frac{10}{3}})$.
However, this bound does not take into account that many of the considered middle levels could be significantly smaller than $O(n^{\frac{4}{3}})$. This gives rise to the following question.
}

\begin{question}\label{Question:projective_levels}
What is the total complexity $\sum_{c\in\mathcal{A}}m(c)$ of all projectively different middle levels?
\end{question}

Any improvement on the bound $O(n^{\frac{10}{3}})$ would immediately improve the bound of the running time of our algorithm.
Further, note that the idea used for the algorithm can also be used to get algorithms for Theorem \ref{Thm:Parallel} and Theorem \ref{Thm:1throughPoint}.

\begin{theorem}
\label{Alg:2DParallel}
Given any three planar point sets $P_1, P_2, P_3$ with a total of $n$ points and a line $\ell$, one can find two lines $\ell_1,\ell_2$ such that $\{\ell_1,\ell_2\}$ simultaneously bisects $P_1, \ldots, P_3$ and $\ell_1$ is parallel to $\ell$ in time $O(n^{\frac{7}{3}})$.
\end{theorem}

\begin{proof}
We know from Theorem \ref{Thm:Parallel} that a solution exists, in which we again may move one line to infinity, namely $\ell_1$.
The duals of the family of lines that are parallel to $\ell$ defines a family of points that are exactly the points on a vertical line~$v$ in the dual, which passes through the dual point $\ell^*$ of the line~$\ell$.
This means that, by fixing $\ell_1$, we place the north pole in a cell intersected by the line~$v$.
As in the previous proof, we consider combinatorially different placements of $\ell_1$ and walk through the respective middle level.
However, the line~$v$ intersects the interior of only $n$ cells,
so we only have to walk along a linear number of middle levels in order to find a solution.
(By the Zone theorem~\cite{ZoneThm}, the cells containing $v$ can be traversed in total~$O(n)$ time.)
This implies the runtime of $O(n^{\frac{7}{3}})$.
\end{proof}

While for Theorem~\ref{Thm:Parallel} the intercept is the only parameter for $\ell_1$ (while the slope is fixed to be the one of~$\ell$), for Theorem~\ref{Thm:1throughPoint} the only parameter for $\ell_1$ is its slope.
The dual of the lines through the given point~$q$ are exactly the points on the dual line $q^*$ of $q$.
If instead of placing the north pole only in cells intersected by the line $v$, we place it only in cells intersected by the line $q^*$, an algorithm for Theorem \ref{Thm:1throughPoint} follows.

\begin{theorem}
\label{Alg:2D1throughPoint}
Given any three planar point sets $P_1, P_2, P_3$ with a total of $n$ points and a point~$q$, one can find two lines $\ell_1,\ell_2$ such that $\{\ell_1,\ell_2\}$ simultaneously bisects $P_1, \ldots, P_3$ and $\ell_1$ goes through $q$ in time $O(n^{\frac{7}{3}})$.
\end{theorem}

We conclude this section by giving an algorithm for our last result in two dimensions, Theorem~\ref{Thm:2throughPoint}.

\begin{theorem}
\label{Alg:2D2throughPoint}
Given any two planar point sets $P_1$ and $P_2$ with a total of $n$ points and a point $q$, one can find two lines $\ell_1,\ell_2$ such that $\{\ell_1,\ell_2\}$ simultaneously bisects $P_1$ and $P_2$ and both $\ell_1$ and $\ell_2$ go through $q$ in time $O(n\log n)$.
\end{theorem}
\begin{proof}
We know from Theorem \ref{Thm:2throughPoint} that a solution exists. Let $\ell$ be any (non-vertical) line through $q$, not passing through any point in $P=P_1\cup P_2$. For any point $p\in P$ that lies below $\ell$, reflect $p$ at $q$. Clearly, this can be done in constant time for each point, so the overall runtime for this step is $O(n)$. Let $P'$ be the point set obtained this way. The crucial observation is that any solution for $P'$ is also a solution for $P$. Order the points in $P'$ along the radial order around $q$ in $O(n\log n)$ time. It now remains to find an interval $I$ in this sequence of points such that $I$ contains exactly half of the points of each point set. As the size of this interval has to be $|P|/2$, there are only linearly many possible intervals, so it is an easy task to find $I$ in linear time. The runtime of the algorithm is therefore dominated by the sorting step.
\end{proof}

\section{Conclusion}
\label{Sec:General}
We have shown that any four mass distributions in the plane can be simultaneously bisected with two lines. We have also shown that we can put additional restrictions on the used lines, at the cost of one or two mass distributions. All of these results are tight in the sense that there is a way to define more mass distributions that cannot be simultaneously bisected with two lines satisfying the imposed restrictions. Also, all the results are accompanied by non-trivial polynomial time algorithms. It remains open whether these algorithms or their runtime analysis can be improved. Also, it would be interesting to find non-trivial lower bounds for the computational complexity of these problems.
%


Going towards more hyperplanes in higher dimensions, we mention the following conjecture by Stefan Langerman:

\begin{conj}
\label{Conj:General}
Any $n\cdot d$ mass distributions in $\mathbb{R}^d$ can be simultaneously bisected by $n$ hyperplanes.
\end{conj}

\begin{remark}
In the conference version of this paper \cite{conference} the conjecture was not attributed to Langerman. We only found out later, that he already made this conjecture years before the conference version was published.
\end{remark}

For $n=1$ this is equivalent to the Ham-sandwich theorem. Theorem \ref{Thm:2Lines} proves this conjecture for the case $d=n=2$. We first observe that the number of mass distributions would be tight:

\begin{obs}
There exists a family of $n\cdot d+1$ mass distributions in $\mathbb{R}^d$ that cannot be simultaneously bisected by $n$ hyperplanes.
\end{obs}

\begin{proof}
Let $P=\{p_1,\ldots,p_{nd+1}\}$ be a finite point set in $\mathbb{R}^d$ in general position (no $d+1$ of them on the same hyperplane). Let $\epsilon$ be the smallest distance of a point to a hyperplane defined by $d$ other points and let $B_i$ be the ball centered at $p_i$ with radius $\frac{\epsilon}{2}$. For each $i\in\{1,\ldots,nd+1\}$ define $\mu_i$ as the volume measure of $B_i$. Note that any hyperplane intersects at most $d$ of the $B_i$'s. On the other hand, for a family of $n$ hyperplanes to bisect $\mu_i$, at least one of them has to intersect $B_i$. Thus, as $n$ hyperplanes can intersect at most $n\cdot d$ different $B_i$'s, there is always at least one $\mu_i$ that is not bisected.
\end{proof}

After the conference version of this paper, the conjecture has been verified for several other values of $n$ and $d$, see \cite{pizza1, pizza2}.

A possible way to prove the conjecture would be to generalize the approach from Section~\ref{Sec:2Cuts} as follows: Consider the $n$ hyperplanes as a highly degenerate algebraic surface of degree $n$, i.e., the zero set of a polynomial of degree $n$ in $d$ variables. Such a polynomial has $k:=\binom{n+d}{d}$ coefficients and can thus be seen as a point on $S^{k-1}$. In particular, we can define $\binom{n+d}{d}-1$ antipodal mappings to $\mathbb{R}$ if we want to apply the Borsuk-Ulam theorem. Using $n\cdot d$ of them to enforce the mass distributions to be bisected, we can still afford $\binom{n+d}{d}-nd-1$ antipodal mappings to enforce the required degeneracies of the surface. There are many conditions known to enforce such degeneracies, but they all require far too many mappings or use mappings that are not antipodal. Nonetheless the following conjecture implies Conjecture~\ref{Conj:General}:

\begin{conj}
\label{Conj:Degenerate}
Let $C$ be the space of coefficients of polynomials of degree $n$ in $d$ variables. Then there exists a family of $\binom{n+d}{d}-nd-1$ antipodal mappings $g_i:C\rightarrow\mathbb{R}$, $i\in\{1,\ldots,\binom{n+d}{d}-nd-1\}$ such that $g_i(c)=0$ for all $i$ implies that the polynomial defined by the coefficients $c$ decomposes into linear factors.
\end{conj}

\section{Acknowledgments}
We would like to thank Stefan Langerman for bringing this problem to our attention, and Johannes Schmitt, Johannes Lengler and Daniel Graf for the fruitful discussions. 
We also want to thank the organizers of CCCG 2017 for a wonderful conference.


\small
\bibliographystyle{abbrv}

\bibliography{refs}




\end{document}